\documentclass[conference]{IEEEtran}
\IEEEoverridecommandlockouts
\usepackage{cite}
\usepackage{amsmath,amsfonts,amssymb}
\usepackage{mathrsfs}
\usepackage{mathtools}
\usepackage{caption}
\usepackage{textcomp}
\usepackage{xcolor}
\usepackage[a4paper, total={184mm,239mm}]{geometry}
\usepackage{tikz}
\usepackage{tikz-cd}
\usepackage{float}
\usepackage{graphicx}
\usepackage{subfigure}
\graphicspath{ {figures/} }
\usepackage{algorithm} 
\usepackage{algpseudocode} 
\usepackage{multirow}
\usepackage{listings}
\usepackage{booktabs}
\usepackage[numbers,sort]{natbib} 	
\usepackage{aliascnt} 
\usepackage{hyperref}
\hypersetup{
pdfusetitle,
bookmarks=true,
pagebackref,
colorlinks,
linkcolor=black,
citecolor=black,
urlcolor=black
}
\usepackage{amsthm}

\theoremstyle{plain}
\newaliascnt{theorem}{dummy}
\newtheorem{theorem}[theorem]{Theorem}
\aliascntresetthe{theorem}

\newaliascnt{proposition}{dummy}

\aliascntresetthe{proposition}

\newaliascnt{corollary}{dummy}

\aliascntresetthe{corollary}

\newaliascnt{lemma}{dummy}

\aliascntresetthe{lemma}

\newaliascnt{conjecture}{dummy}

\aliascntresetthe{conjecture}

\theoremstyle{definition}
\newaliascnt{definition}{dummy}

\aliascntresetthe{definition}

\newaliascnt{example}{dummy}

\aliascntresetthe{example}

\theoremstyle{remark}
\newaliascnt{remark}{dummy}

\aliascntresetthe{remark}

\def\BibTeX{{\rm B\kern-.05em{\sc i\kern-.025em b}\kern-.08em
    T\kern-.1667em\lower.7ex\hbox{E}\kern-.125emX}}
\begin{document}

\title{SmaRTLy: RTL Optimization with Logic Inferencing and Structural Rebuilding\\
}

\author{\IEEEauthorblockN{Chengxi Li}
\IEEEauthorblockA{\textit{Dept. of CSE} \\
\textit{The Chinese University of Hong Kong}\\
Hong Kong, China \\
cxli23@cse.cuhk.edu.hk}
\and
\IEEEauthorblockN{Yang Sun}
\IEEEauthorblockA{\textit{Dept. of CSE} \\
\textit{The Chinese University of Hong Kong}\\
Hong Kong, China \\
ysun22@cse.cuhk.edu.hk}
\and
\IEEEauthorblockN{Lei Chen}
\IEEEauthorblockA{\textit{Noah's Ark Lab} \\
\textit{Huawei}\\
Hong Kong, China \\
lc.leichen@huawei.com}
\and
\IEEEauthorblockN{Yiwen Wang}
\IEEEauthorblockA{\textit{Noah's Ark Lab} \\
\textit{Huawei}\\
Hong Kong, China \\
wangyiwen19@huawei.com}
\and
\IEEEauthorblockN{Mingxuan Yuan}
\IEEEauthorblockA{\textit{Noah's Ark Lab} \\
\textit{Huawei}\\
Hong Kong, China \\
yuan.mingxuan@huawei.com}
\and
\IEEEauthorblockN{Evangeline F.Y. Young}
\IEEEauthorblockA{\textit{Dept. of CSE} \\
\textit{The Chinese University of Hong Kong}\\
Hong Kong, China \\
fyyoung@cse.cuhk.edu.hk}
}

\maketitle

\begin{abstract}
This paper proposes smaRTLy: a new optimization technique for multiplexers in Register-Transfer Level (RTL) logic synthesis. Multiplexer trees are very common in RTL designs, and traditional tools like Yosys optimize them by traversing the tree and monitoring control port values. 
However, this method does not fully exploit the intrinsic logical relationships among signals or the potential for structural optimization.
To address these limitations, we develop innovative strategies to remove redundant multiplexer trees and restructure the remaining ones, significantly reducing the overall gate count.
We evaluate smaRTLy on the IWLS-2005 and RISC-V benchmarks, achieving an additional 8.95\% reduction in AIG area compared to Yosys.
We also evaluate smaRTLy on an industrial benchmark in the scale of millions of gates, results show that smaRTLy can remove 47.2\% more AIG area than Yosys.
These results demonstrate the effectiveness of our logic inferencing and structural rebuilding techniques in enhancing the RTL optimization process, leading to more efficient hardware designs.
\end{abstract}

\begin{IEEEkeywords}
Multiplexer, Optimization, RTL, Synthesis, Structure
\end{IEEEkeywords}

\section{Introduction}
Logic optimization is an important part in logic synthesis.
Multiplexers, often abbreviated as MUX, are common and important in digital circuits.
In communication systems, they enable the transmission of multiple signals over a single communication line, thereby optimizing the use of bandwidth. 
In data routing, multiplexers facilitate efficient transfer of data between different parts of a computer system. 
Additionally, they are used in implementing Boolean functions and in the design of complex digital circuits\cite{macko2013managing,siebert2013delay}.
Therefore, optimization on multiplexers plays an important role in logic optimization\cite{lin2022novelrewrite,liu2023rethinking,liu2024unified,sun2024massively}.

A common 2-to-1 MUX consists of two data input ports ($A$ and $B$), a single-bit control port ($S$), and an output port ($Y$). 
More complex n-to-1 MUX can be formed by combining multiple 2-to-1 MUXs.
Multiple MUX can form a tree in a combinational circuit to represent \textit{if-else} and \textit{case} statements in Hardware Description Language (HDL), which is called muxtree\cite{metzgen2005multiplexer}.

Many logic synthesis tools include optimization pass on muxtree\cite{wang2023optimization, kohutka2014faster, pivstek2010optimization}.
Yosys\cite{wolf2016yosys}, a famous open-source Verilog HDL synthesis tool, primarily employs the \textit{opt\_muxtree} pass to optimize multiplexer trees. This pass analyzes control signals to identify and remove never-active branches by traversing the multiplexer trees and monitoring the values of visited control ports.
A MUX will be removed if it shares the same control signal with visited MUXs.
Specifically, it will optimize muxtree in the following two situations. First, when the value of the control port is determined in an ancestor MUX, as shown in \autoref{fig:origin_1}, the muxtree of $Y=S\ ?\ (S \ ?\ A:B):C$ can be optimized into $Y=S\ ?\ A:C$. Second, when the value of the data port is determined in an ancestor MUX, as shown in \autoref{fig:origin_2}, the muxtree of $Y=S\ ?\ (A \ ?\ S:B):C$ can be optimized into $Y=S\ ?\ (A \ ?\ 1:B):C$.

\begin{figure}[tb]
	\centering
	\includegraphics[width=0.38\textwidth]{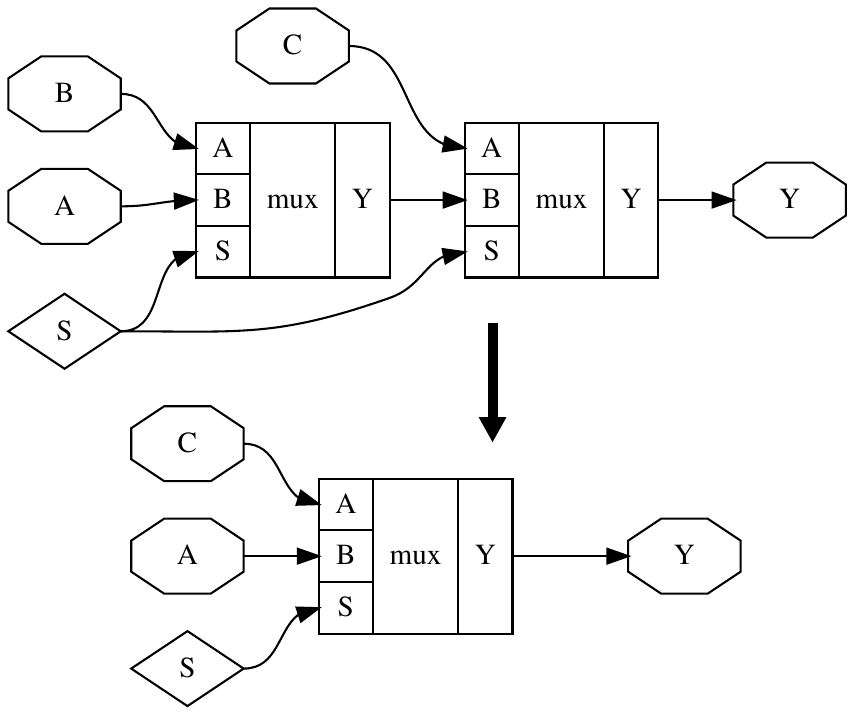}
	\caption{The multiplexer with the same control as the ancestor multiplexer can be optimized.}
    \label{fig:origin_1}
\end{figure}

\begin{figure}[tb]
	\centering
	\includegraphics[width=0.38\textwidth]{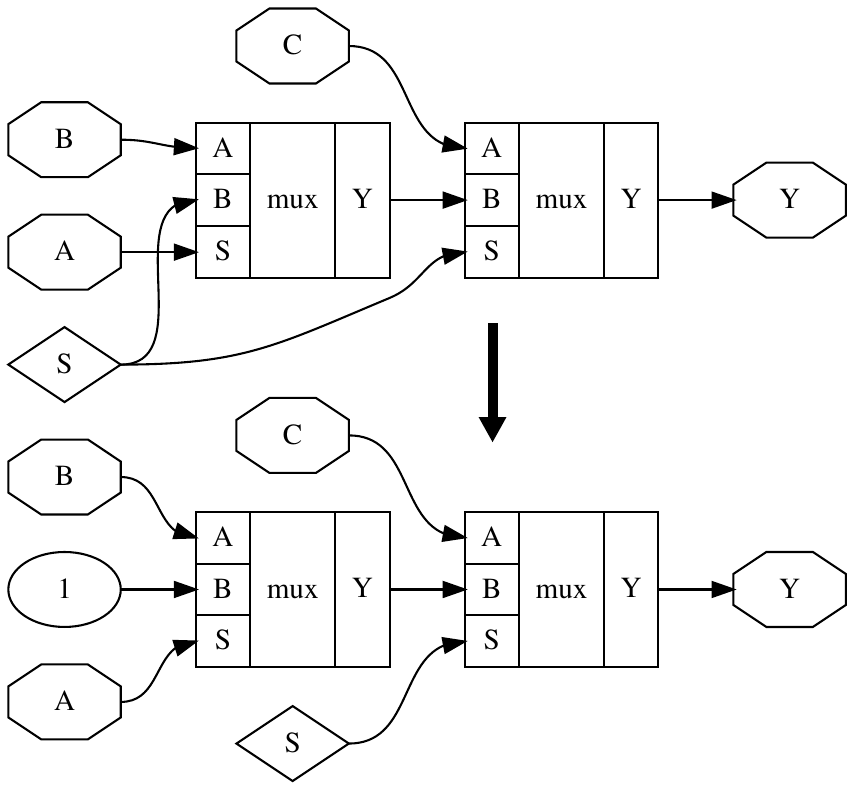}
	\caption{The data port with the same control signals as the ancestor multiplexer can be optimized.}
    \label{fig:origin_2}
\end{figure}

The method is feasible and efficient. In most cases, it can remove abundant gates and wires.
However, the method has the following disadvantages:

\begin{itemize}
    \item Limited logical analysis: The method ignores intrinsic logical relationships between signals. Traditional optimization occurs only when MUXs share identical signals, failing to optimize when control ports differ but are dependent. This limitation is illustrated in \autoref{fig:simulation}, the muxtree of $Y=S\ ?\ ((S|R) \ ?\ A:B):C$ can be optimized into $Y=S\ ?\ A:C$, because $S|R=1$ when $S=1$. However, this can be ignored in traditional methods.
    \item Structural inflexibility: Many approaches merely remove unused branches of a muxtree without changing its structure. However, many muxtrees can be optimized by rebuilding its structure for improved performance. A more detailed explanation will be given in \autoref{sec:rebuild}.
\end{itemize}

\begin{figure}[tb]
	\centering
	\includegraphics[width=0.5\textwidth]{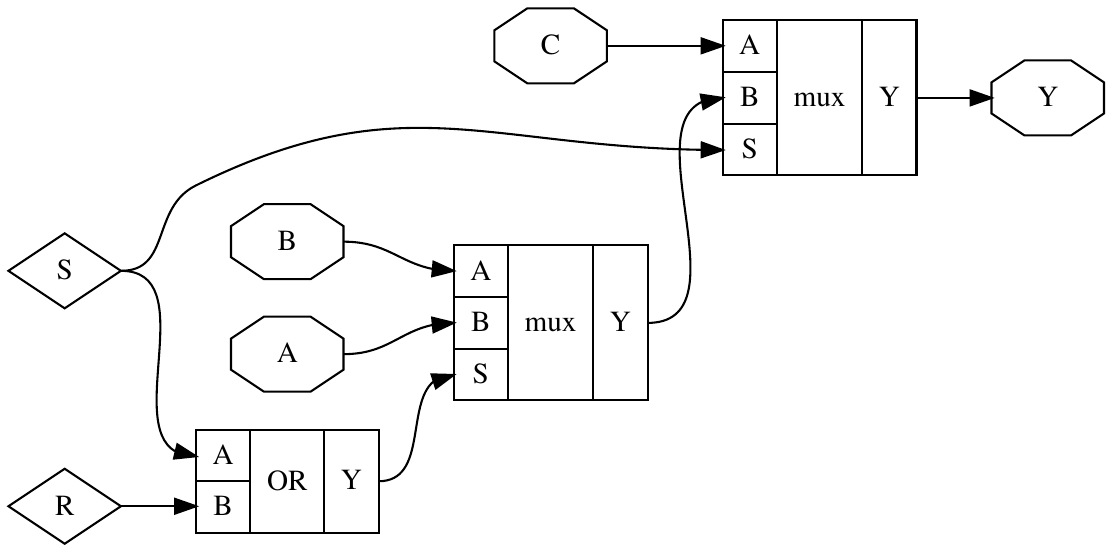}
	\caption{The muxtree with different and logically dependent control signals.}
    \label{fig:simulation}
\end{figure}

To address these problems, we propose smaRTLy, which contains two optimizations, SAT-based redundancy elimination on sub-graphs and muxtree restructuring.

For SAT-based redundancy elimination, we focus on optimizing multiplexers with different control signals. Our approach involves analyzing the control port values and using an SAT solver to identify and eliminate redundant multiplexers. To manage complexity, we constrain the analysis to a sub-graph of the circuit and minimize the size of the sub-graph without sacrificing correctness.

For muxtree restructuring, we target at muxtrees with inefficient structures to eliminate more gates.
We identify these structures, assess their efficiency, and use an Algebraic Decision Diagram (ADD) to rebuild the muxtree with fewer multiplexers~\cite{bahar1997algebric}.
This process involves reassigning control signals and outputs to optimize the overall design.

SmaRTLy were evaluated on 10 cases from IWLS-2005~\cite{albrecht2005iwls} and RISC-V~\cite{waterman2014risc} benchmarks compatible with Yosys. Results show an additional 8.95\% reduction in AIG areas compared to Yosys.
The methods can be combined for enhanced performance and have the potential to integrate with other optimization techniques.
Besides, smaRTLy can perform much better in an industrial benchmark, where the average size reaches millions of gates, and smaRTLy can remove 47.2\% more AIG area than Yosys.
This shows the effectiveness of smaRTLy on practical cases.

SmaRTLy offers improvements in circuit optimization, potentially leading to better PPA (Power, Performance, Area). As digital circuits grow more complex, such optimization techniques become increasingly valuable in managing design complexity and improving the overall system performance.
Our contributions in this work are summarized below:
\begin{itemize}
    \item We propose an SAT-based redundancy elimination to capture the logical relationship between signals.
    \item A novel method was developed that can greatly reduce the size of an RTL sub-graph.
    \item We propose an efficient muxtree restructuring algorithm using heuristics based on Algebraic Decision Diagram (ADD).
    \item Experimental results demonstrate improved optimization compared to traditional RTL optimizations in Yosys, highlighting the effectiveness of smaRTLy.
\end{itemize}

The rest of the paper is organized as follows.
Section 2 presents SAT-based redundancy elimination.
Section 3 presents Muxtree Restructuring, followed by experiment results in Section 4. 
Finally, Section 5 concludes the paper.
\section{SAT-based Redundancy Elimination}





The frequent use of \textit{if-else} and \textit{case} statements in RTL-level code often results in a significant number of multiplexers (MUXs) being generated during logical synthesis.
This is a natural consequence of designing complex circuits, which typically involve multiple layers of nesting and numerous conditional statements. 
As a result, muxtrees become a prevalent feature in today's designs. 
While muxtrees are essential for implementing intricate logic, they can also lead to inefficiencies, particularly when they contain redundant MUXs. 
As the design logic becomes increasingly intricate, it is common to find redundant MUXs within these muxtrees.

A MUX is considered redundant if the value of its control port is already determined when traversing from the root to the MUX itself. 
In other words, the control signal's value can be deduced from known signals along the path. 
Similarly, optimization can be applied to the data port by analyzing these signals.
Current tools like Yosys primarily optimize MUXs with identical control signals. 
However, this approach overlooks many scenarios in digital circuits where control signals are interconnected through logic gates. 
For instance, complex logical relationships between control signals are common, as illustrated in \autoref{fig:simulation}. 

To address this limitation, we propose a SAT-based redundancy elimination method on the sub-graph to uncover deeper logical relationships. SAT (Boolean Satisfiability) solvers are powerful tools for determining whether there exists an assignment of variables that satisfies a given Boolean formula. By simplifying the sub-graph based on signal locations and applying SAT-based techniques on the sub-graph to the MUX optimization problem, it becomes possible to identify and eliminate redundancies that are not detectable through simpler methods like simulation.

SmaRTLy begins by constructing a sub-graph during the traversal of the muxtree. 
When a new MUX is encountered, all logical gates within a specified distance \textit{k} from the control port are incorporated, provided they are not already included in the sub-graph. 
This approach allows for a more comprehensive analysis of the logical relationships between signals.
However, if \textit{k} is large, the size of the sub-graph will be too large for the SAT solver to take too much time; if \textit{k} is small, the sub-graph will not contain enough nodes for the SAT solver to infer the value of the target.
This is because the sub-graph contains too many irrelevant gates and signals, making the sub-graph bloated.
To keep the sub-graph manageable, smaRTLy only adds potential signals whose values might be affected by known signals. 
We propose the scenarios in which the signal $S$ can impact $T$:

\begin{theorem}
In the RTL level circuit, a signal $S$ is possible to affect the value of signal $T$, only if:
\begin{itemize}
    \item $S$ is an ancestor of $T$, or
    \item $T$ is an ancestor of $S$, or
    \item $S$ and $T$ have a common ancestor.
\end{itemize}
\label{thm:subgraph}
\end{theorem}

The proof is not complicated; we only need to consider the scenarios not in \autoref{thm:subgraph}:

\begin{proof}
If the relationship between signal $S$ and signal $T$ is not in the three scenarios of \autoref{thm:subgraph}, then:
\begin{itemize}
    \item There's no path from $S$ to $T$ (i.e. $S$ and $T$ are not connected), or
    \item $S$ and $T$ only have a common descendant.
\end{itemize}
In both scenarios, $S$ is impossible to affect the value of $T$.
\end{proof}

\begin{figure}[tb]
	\centering
	\includegraphics[width=0.25\textwidth]{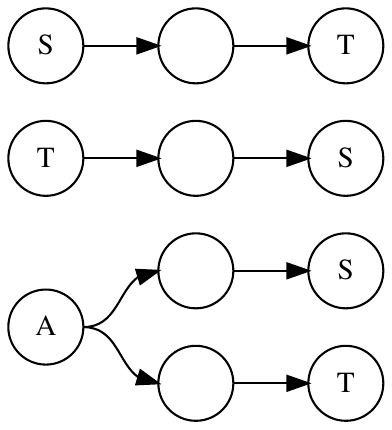}
	\caption{The three situations where S(T) can affect T(S).}
    \label{fig:subgraph}
\end{figure}

The three scenarios are shown in \autoref{fig:subgraph}. We can find that $T$ is also possible to affect $S$:

\begin{theorem}
If signal $S$ is possible to affect the value of signal $T$, then $T$ is also possible to affect the value of $S$.
\label{thm:subgrap2}
\end{theorem}

\begin{proof}
When $S$ is possible to affect the value of $T$:
\begin{itemize}
    \item If $S$ is an ancestor of $T$, then $T$ is a descendant of $S$.
    \item If $T$ is an ancestor of $S$, then $S$ is a descendant of $T$.
    \item If $S$ and $T$ have a common ancestor, then $T$ and $S$ have a common ancestor
\end{itemize}
So $T$ is also possible to affect the value of $S$.
\end{proof}

From~\autoref{thm:subgraph} and \autoref{thm:subgrap2}, we can see that these signals can be partitioned into groups such that signals in the same group will affect each other while those from different groups will not. This allows us to disregard gates without potential signals in effect, streamlining our analysis. 
The method can dismiss about 80\% gates in the sub-graph, greatly accelerating the inference of the SAT solver.
Additionally, smaRTLy excludes sequential gates to maintain the sub-graph as a Directed Acyclic Graph (DAG).

To enhance the efficiency of the SAT solver, smaRTLy then employs inference rules. 
Considering that the logical relationships are often not overly complex, as illustrated in \autoref{fig:simulation}, straightforward inferences can help reduce unknown signals. 
smaRTLy apply the inference rules shown in \autoref{tab:inference_rules} to the known value signals. 
If a condition matches, the corresponding signal in the result becomes a new known value signal.


\begin{table}[tb]
\caption{Inference rules for \textit{OR} cells.}
\begin{center}
    \normalsize
    \renewcommand\arraystretch{1.5}
    \begin{tabular}{c|c}
       Condition & Result \\ \hline
       $a=true$ & $a|b = true$ \\
       $b=true$ & $a|b=true$ \\
       $a=b=false$ & $a|b = false$ \\
       $a|b=false$ & $a=b=false$ \\
       $a|b=true\ a=false$ & $b=true$ \\
       $a|b=true\ b=false$ & $a=true$ \\
    \end{tabular}
\label{tab:inference_rules}
\end{center}
\end{table}


Lastly, smaRTLy performs both simulation and SAT solver techniques on the sub-graph to determine signal values. 
To optimize efficiency, we choose between these methods based on the number of inputs requiring assignment. 
For a smaller number of inputs, simulation is more efficient, while the SAT solver is better suited for handling larger sets of inputs.
The determination of a signal $S$ is achieved by checking if either $SAT(S=0)=false$ or $SAT(S=1)=false$, indicating that the signal's value is fixed. 
In smaRTLy, we employ MiniSAT \cite{sorensson2005minisat}, a robust and efficient SAT solver, to handle complex logical evaluations. While there are more recent SAT solvers that might offer improved performance in certain scenarios, MiniSAT remains a popular choice due to its balance of efficiency and ease of integration.

To prevent excessive computational overhead, we introduce a threshold for the number of inputs. 
If the number of inputs in the sub-graph becomes excessively large, they may choose to forgo the SAT process. 
This decision is made to prevent the optimization process from becoming a bottleneck in the overall circuit synthesis workflow.
This selective approach ensures that our method remains practical and responsive to varying complexities within the sub-graph.
\section{Muxtree Restructuring}
\label{sec:rebuild}


Yosys focuses on removing redundant MUXs from muxtrees without changing their structure and control ports. 
However, reconstructing these muxtrees into more efficient structures makes it possible to reduce the number of MUXs and connected gates further. 
This structural change can minimize resource usage and decrease circuit area and delay, ultimately enhancing the circuit's performance and sustainability.

\begin{lstlisting}[caption=HDL code for a case statement, label=list:case_statement, float]
case(S)
    2'b00:   Y = p0;
    2'b01:   Y = p1;
    2'b10:   Y = p2;
    default: Y = p3;
endcase
\end{lstlisting}

\begin{figure}[h!]
	\centering
	\includegraphics[width=0.5\textwidth]{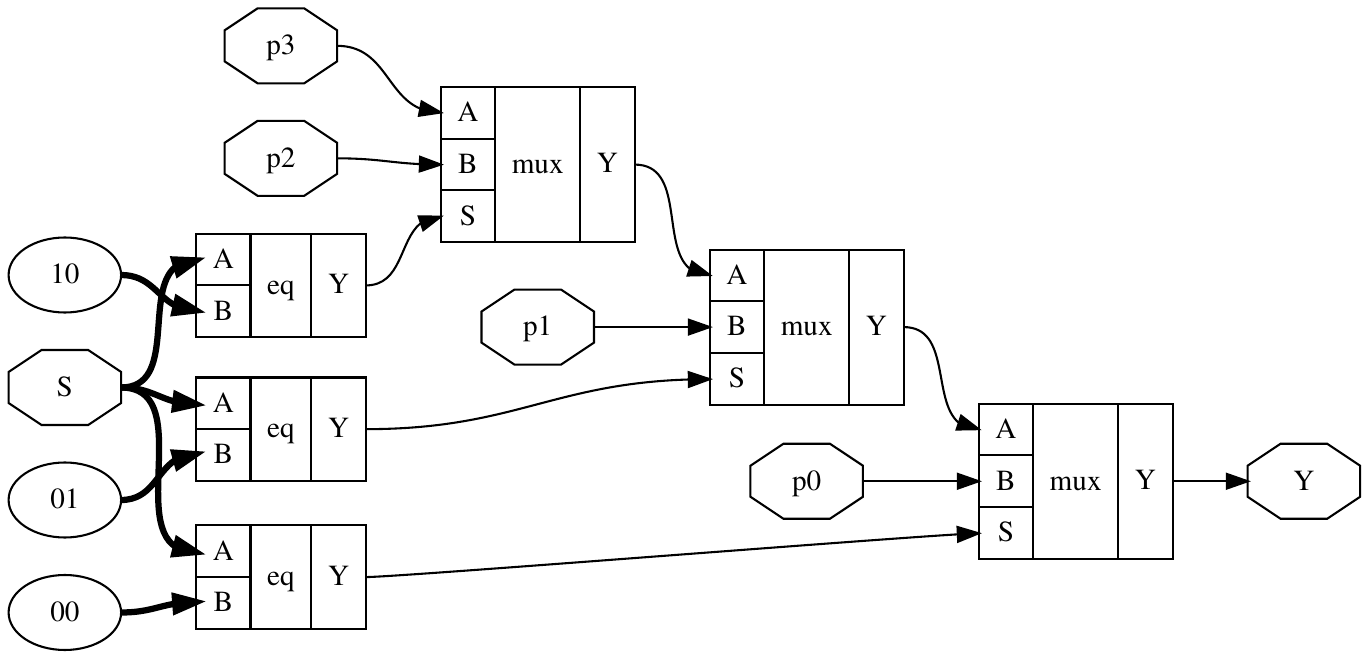}
	\caption{A muxtree for the $case$ statement in \autoref{list:case_statement}. The muxtree is a chain.}
    \label{fig:muxtree_case}
\end{figure}

\begin{figure}[h!]
	\centering
	\includegraphics[width=0.5\textwidth]{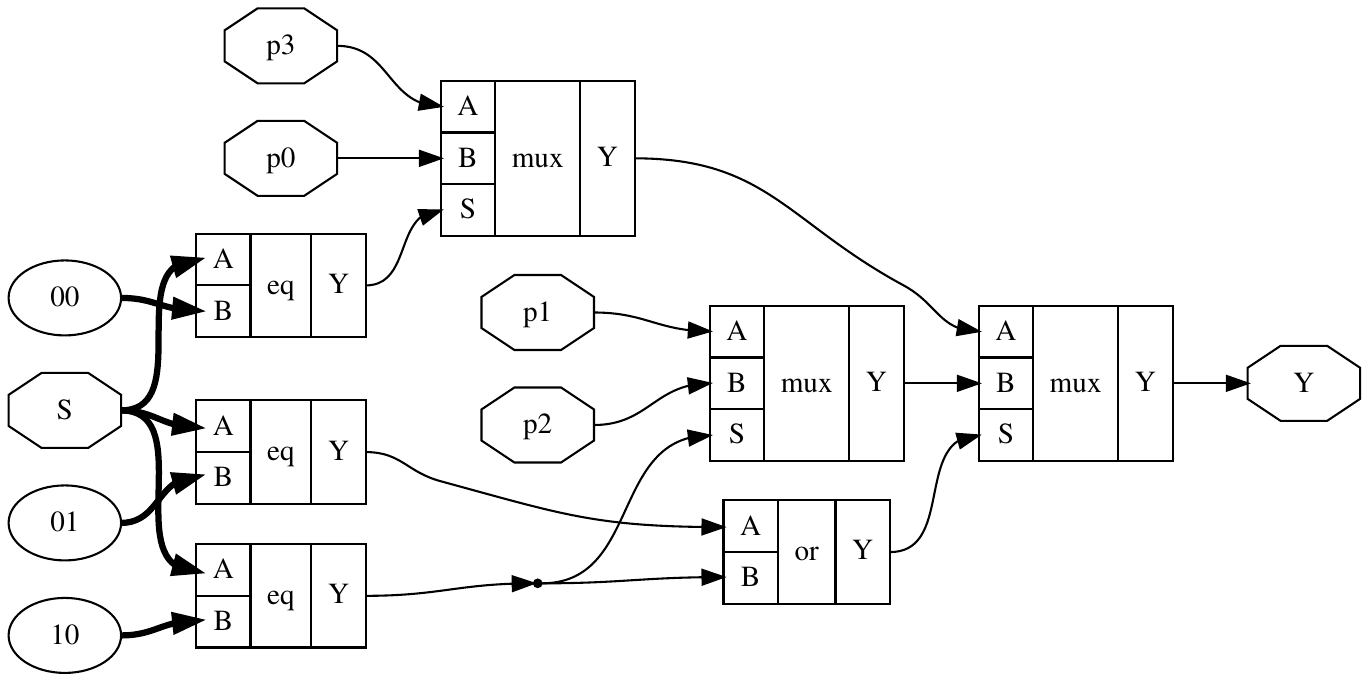}
	\caption{Another muxtree for the $case$ statement in \autoref{list:case_statement}. The muxtree is a full binary tree.}
    \label{fig:muxtree_case2}
\end{figure}

In general, smaRTLy will rebuild muxtrees with a single control signal into a better structure.
For example, the HDL code in \autoref{list:case_statement} is represented as a muxtree with three $eq$ gates and three MUX gates, as illustrated in \autoref{fig:muxtree_case} and \autoref{fig:muxtree_case2}. 
The muxtree in \autoref{fig:muxtree_case} forms a chain, which is inherently inefficient due to its sequential nature. 
In contrast, the muxtree in \autoref{fig:muxtree_case2} is designed as a full binary tree with an additional $or$ gate. 
While this structure improves efficiency, it fails to account for the logical relationships between control ports.
Their control signals are all signal $S$, so it can be formed as a more efficient structure
By revising the structure and control ports, the muxtree can be optimized to utilize only three MUX gates, as demonstrated in \autoref{fig:stucture}. 
This approach not only simplifies the design but also enhances its efficiency by eliminating redundant components. 
The three $eq$ gates are disconnected from the muxtree and will be removed if they do not connect to any other gates, further simplifying the circuit. 

\begin{figure}[h!]
	\centering
	\includegraphics[width=0.4\textwidth]{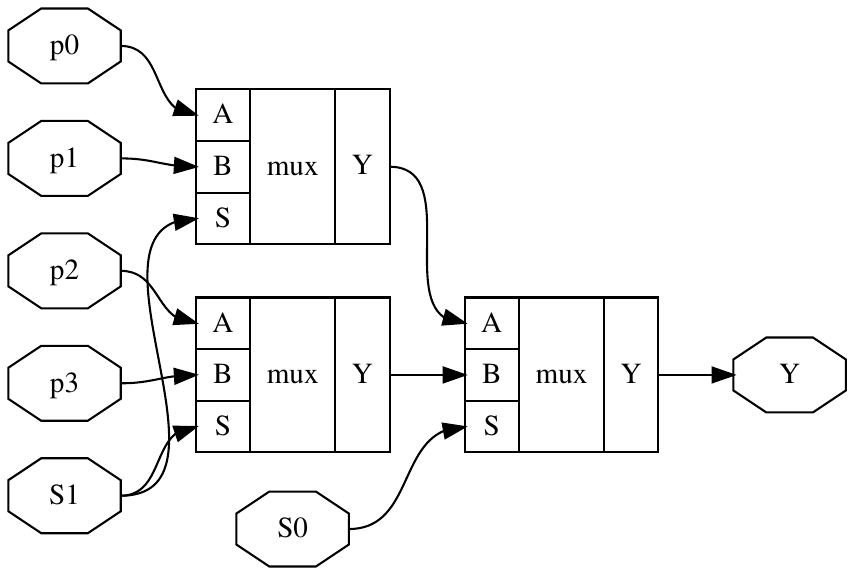}
	\caption{A rebuilt muxtree connected with fewer gates.}
    \label{fig:stucture}
\end{figure}



\begin{algorithm}
	\caption{Muxtree Restructuring} 
	\begin{algorithmic}[1]
		\For {$cell \ in \ \{ muxtree\_roots\}$}
                \If {$OnlyEq(cell) \ and \ SingleCtrl(cell)$}
                    \State $Assignment \leftarrow ADD(cell)$
                    \State $RemovedEq \leftarrow CountRemoved(cell)$
                    \State $height \leftarrow cell.ctrl.width$
                    \State $width \leftarrow cell.width$
                    \If {$Check$($Assignment$, $RemovedEq$, $height$, $width$)}
                        \State $Rebuild(cell, Assignment)$
                        \State $RemoveUnusedCell()$ \Comment{Implemented in other pass}
                    \EndIf
                \EndIf
		\EndFor
	\end{algorithmic} 
    \label{alg:rebuild}
\end{algorithm}

The algorithm of Muxtree Restructuring is shown in Algorithm \autoref{alg:rebuild}.
smaRTLy first identifies all muxtrees with structures that can be optimized. 
In this paper, we focus on muxtrees generated from \textit{case} statements, as they typically use \textit{eq} gates or \textit{logic\_not} as control ports (\textit{logic\_not} is a special \textit{eq}) and share the same control signal.

Next, smaRTLy assigns signals to the control ports of each MUX. 
A well-chosen assignment enables us to rebuild the muxtree with fewer MUXs. 
Conversely, a poor assignment can lead to excessive MUX usage, increasing complexity and resource consumption. 
For example, with the HDL code in \autoref{list:assignment}, a good assignment (e.g., assigning ($S_2$ to $S_0$) results in 3 MUXs, whereas a poor assignment ($S_0$ to $S_2$) results in 7 MUXs.
So smaRTLy collects all the inputs of control ports and corresponding outputs, representing them as an Algebraic Decision Diagram (ADD). 
ADD is a generalization of Binary Decision Diagram (BDD) from $\{0,1\}$ output sets to arbitrary finite output sets. 
For example, the output set of \autoref{fig:stucture} is $\{p_0,p_1,p_2,p_3\}$ and \autoref{fig:stucture} is the optimal solution in this case.
As BDD is a well-known NP-C problem, so is ADD.
Considering the trade-off between performance and efficiency, we use a simple heuristic algorithm: for each MUX, smaRTLy selects the signal that minimizes the total types of terminal nodes of the left and right children.
Take \autoref{list:assignment} as an example, if we select $S_2$ as the control signal of the first MUX, the number of types of terminal nodes is 4 (left: $\{p_1, p_2, p_3\}$, right: $\{p_0\}$); but if we select $S_0$, the number is 6 (left: $\{p_0, p_1, p_3\}$, right: $\{p_0, p_1, p_2\}$).
In experiments, the algorithm can obtain the optimal solution or close to the optimal solution in most cases. This is because the ADD of actual circuits is usually simple and does not require complex algorithms.

\begin{lstlisting}[caption=HDL code for another case statement, label=list:assignment, float]
case(S)
    3'b1zz:   Y = p0;
    3'b01z:   Y = p1;
    3'b001:   Y = p2;
    default:  Y = p3;
endcase
\end{lstlisting}



Next, smaRTLy will determine whether to rebuild the muxtree. 
If we roughly rebuild all muxtrees that meet the line 2 conditions of Algorithm \autoref{alg:rebuild}, the optimization effect is often poor and may even deteriorate the circuit.
This is because while smaRTLy can disconnect the $eq$ gates from the muxtree, they may remain in the circuit if connected to other gates.
This connectivity can affect overall circuit performance and complexity. 
Additionally, a poor assignment or particularly challenging cases might necessitate extra MUXs during reconstruction, which can increase resource usage and circuit area and delay.

To make an informed decision on whether to rebuild the muxtree, we consider several critical factors. 
First, smaRTLy assesses the potential for removing unnecessary gates and only handles the potential ones. 
Then, smaRTLy evaluates the height of the rebuilt muxtree, as the height determines the level of complexity in the muxtree structure. 
The width of the MUX is another important consideration, as a MUX gate is divided into several single-bit gates after technology mapping. 
smaRTLy also examines the number of additional MUXs needed to maintain the functional equivalence.
By taking these factors into account, smaRTLy can decide whether the benefits of muxtree restructuring outweigh the costs. 
If rebuilding is deemed advantageous, smaRTLy proceed with restructuring according to the optimal assignment. 
Finally, smaRTLy removes any redundant gates that are no longer connected to the muxtree.

This careful process ensures that the circuit remains optimized, reducing unnecessary components and enhancing overall performance.
\section{Experiment}
\subsection{Public Benchmarks}



SmaRTLy is implemented using C/C++.
All the experiments are implemented on Ubuntu 20.04.6 LTS.
The CPU is Intel(R) Xeon(R) Silver 4114 CPU @ 2.20GHz.
The version of GCC is 9.4.0.

We select the top 10 largest circuits written in Verilog from open-source datasets, including IWLS-2005 and RISC-V.
For fairness of comparison, we focus on the AIG area, specifically the number of AND gates in the optimized circuit. 
We exclude Flip-Flop gates from consideration to maintain consistency. 
We replaced the \textit{opt\_muxtree} pass in Yosys with smaRTLy and used the built-in command \textit{aigmap} in Yosys to convert netlists into AIG.
All the results generated by our program passed equivalence checking.

The results are presented in \autoref{tab:result}.
For each case, we first show its original AIG area. Then we show the AIG area of the case optimized by Yosys and smaRTLy, and the last column shows the percentage of area reduction of smaRTLy compared to Yosys, with positive numbers representing improvements.
In our benchmark, the range of AIG area is from 23709 to 10836722 and the average is 1415259.6. 
On average, smaRTLy achieved an extra reduction of 8.95\%, which signifies a substantial improvement in logic optimization, as detailed in \autoref{tab:result_split}.
This reduction is not very large because Yosys is optimized very well on the public dataset, leaving little room for optimization for smaRTLy. 
The average optimization of Yosys is 55\% and in some cases, it can achieve 94\%.
Therefore, this reduction is enough to demonstrate the effectiveness of smaRTLy in optimizing circuit designs, leading to more efficient and resource-saving implementations.

\begin{table}[tb]
\small
\centering
\caption{Comparison of AIG areas between Yosys and smaRTLy on benchmark circuits. The first four columns of each row are the case name, the original AIG area, the AIG area after Yosys and smaRTLy optimization. The last column is the proportion of area reduced by smaRTLy compared to Yosys.}
\label{tab:result}

\scalebox{1.0}{
\begin{tabular}{lrrrr}
\toprule
        Case & Original & Yosys & smaRTLy & Ratio \\
\midrule
top\_cache\_axi                & 10836722                           & 1301437                         & 977118   & 24.92\% \\ 
pci\_bridge32                  & 61847                              & 47411                           & 44369    & 6.42\%  \\ 
wb\_conmax                     & 336039                             & 123659                          & 89290    & 27.79\% \\ 
mem\_ctrl                      & 1118764                            & 65785                           & 65437    & 0.53\%  \\ 
wb\_dma                        & 592158                             & 74697                           & 64322    & 13.89\% \\ 
tv80                           & 772802                             & 46137                           & 45070    & 2.31\%  \\ 
usb\_funct                     & 76287                              & 40571                           & 39095    & 3.64\%  \\ 
ethernet                       & 124127                             & 113507                          & 112202   & 1.15\%  \\ 
riscv                          & 210141                             & 121280                          & 118689   & 2.14\%  \\ 
ac97\_ctrl                     & 23709                              & 23173                           & 21622    & 6.69\%  \\
\midrule
\textbf{Average}               & 1415259.6                          & 195765.7                        & 157721.4 & 8.95\%  \\ 
\bottomrule
\end{tabular}
}
\end{table}


We also present the results of two optimizations in smaRTLy. 
On average, SAT-based redundancy elimination (shown as SAT) achieves a reduction of 3.57\%, while muxtree rebuilding (shown as Rebuild) achieves 4.39\%. 
However, the effectiveness of these optimizations can vary significantly between cases.

For instance, in the case of \textit{top\_cache\_axi}, Rebuild reduces the AIG area by 24.91\%, whereas SAT achieves only a 0.01\% reduction. Conversely, for \textit{wb\_conmax}, SAT achieves a 19.05\% reduction, while Rebuild only manages 4.65\%.
These differences arise from the specific characteristics of each instance. Some examples, with complex data path logic, benefit more from SAT optimization. Others, with numerous case statements, are better suited to Rebuild optimization.


\begin{table}[tb]
\small
\centering
\caption{Reduction in AIG area by individual methods and their combined effect.}
\label{tab:result_split}

\scalebox{1.2}{
\begin{tabular}{lrrr}
\toprule
        Case & SAT & Rebuild & Full \\
\midrule
        top\_cache\_axi & 0.01\% & 24.91\% & 24.92\% \\
        pci\_bridge32 & 0.71\% & 2.01\% & 6.42\% \\
        wb\_conmax & 19.05\% & 4.65\% & 27.79\%  \\
        mem\_ctrl & 0.12\% & 0.47\% & 0.53\%  \\
        wb\_dma & 11.52\% & 0.80\% & 13.89\%  \\
        tv80 & 0.71\% & 1.61\% & 2.31\%  \\
        usb\_funct & 1.60\% & 1.69\% & 3.64\%  \\
        ethernet & 0.49\% & 0.48\% & 1.15\%  \\
        riscv & 0.17\% & 1.97\% & 2.14\%  \\
        ac97\_ctrl & 1.34\% & 5.36\% & 6.69\% \\
\midrule
        \textbf{Average} & 3.57\% & 4.39\% & 8.95\% \\
\bottomrule
\end{tabular}
}
\end{table}

In addition, Full optimization reduces more areas than SAT and Rebuild optimization combined in most cases. 
This is because SAT and Rebuild optimizations work together to reduce more areas.
For example, Rebuild optimization can reduce the height of muxtrees and simplify the control port, which will make the sub-graph smaller in SAT optimization.
In logical synthesis, combining multiple optimizations typically yields better results.

\subsection{Industrial Benchmarks}
We also tested smaRTLy on an industrial benchmark.
Due to confidentiality, we will only show the order of magnitude of the benchmark and the improvement of smaRTLy over Yosys.
The average AIG area size of this benchmark is in the millions, of which 37.5\% of the test points have an area of more than one million AIG nodes.

SmaRTLy can remove 47.2\% more AIG areas than Yosys.
The result far exceeds the performance of smaRTLy on IWLS-2005 and RISC-V.
This is because Yosys performs poorly on this industrial benchmark. In some cases, there is almost no optimization effect.
Moreover, the selection circuits are more common in the industrial dataset, so the proportion of MUX gates and PMUX gates is higher.
The good performance highlights the efficiency and effectiveness of smaRTLy in processing large-scale industrial applications.
\section{Conclusions}
In this paper, we propose smaRTLy on multiplexers in the logic synthesis of RTL-level circuits.
smaRTLy can deal with more complex relationships between the control ports and remove more gates by rebuilding the muxtree with a better structure.
They demonstrated good optimization performance on the public benchmark and industrial benchmark, reducing 8.95\% and 47.2\% more AIG area than Yosys respectively.
In addition, the two optimizations work well together to achieve better performance and have the potential to work with more optimizations in the future.

\bibliographystyle{plain}
\bibliography{reference}

\end{document}